\documentclass{article}
\usepackage[T1]{fontenc}
\usepackage{lmodern}
\usepackage[margin=2cm]{geometry}
\usepackage{amsmath,amssymb,amsfonts}
\usepackage{fancyhdr}
\usepackage[utf8]{inputenc}  
\usepackage{color}
\usepackage{amsmath,amssymb,amsthm,mathrsfs,amsfonts,dsfont,bm} 
\usepackage{authblk}
\usepackage[english]{babel}
\usepackage[autolanguage]{numprint}
\usepackage{graphicx}
\usepackage{verbatim}
\usepackage{comment}
\usepackage{authblk}
\providecommand{\keywords}[1]{\textbf{\textit{key words : }} #1}

\newcommand{\R}{\mathbb{R}}

\newcounter{exocount}
\newcounter{questcount}

\newcommand{\Rr}{\mathbb{R}}

\renewcommand {\epsilon}{\varepsilon}

\renewcommand {\leq}{\leqslant}
\renewcommand {\geq}{\geqslant}

\theoremstyle{definition}
\theoremstyle{plain}

\newtheorem{exple}{Example}[section]
\newtheorem{thm}{Theorem}[section]

\newtheorem{lemma}[thm]{Lemma}
\newtheorem{rque}{Remark} [section]
\newtheorem{cor}[thm]{Corollary}

\title{Analytical solution of $k$th price auction}

\author[a]{Martin Mihelich\thanks{mmihelich@openpricer.com}}
\author[b]{Yan Shu\thanks{ yan.shu@walnut.ai} }
\affil[a]{Open Pricer}
\affil[b]{Walnut Algorithms}

\begin{document}
\maketitle

\begin{abstract}
We provide an exact analytical solution of the Nash equilibrium for the $k$th price auction by using inverse of distribution functions. As applications, we identify the unique symmetric equilibrium where the valuations have
polynomial distribution, fat tail distribution and exponential distributions.
\end{abstract}

\keywords
Vickrey auctions, $k$th price auctions, inverse distribution functions

\section{Introduction}

In a $k$th price auction with $k$ or more bidders, the highest bidder wins the bid and pays the $k$th highest bid as price. The $k$th price auction has been studied by many researchers in recent years. Readers can refer to \cite{PK99, KRISHNA20101,EM04, WILSON92} for related literature. The Revenue Equivalence Theorem (RET) (see \cite{Roger1981}, \cite{RS1981}) can be used to characterize equilibrium strategies of $k$th price auction. Monderer and Tennenholtz \cite{monderer2000k} proved the uniqueness of the  equilibrium strategies in $k$th price auctions for $k=3$. Under some regularity assumptions, they also provided a characterization equation of the equilibrium bid function (see theorem \ref{thm:MT} below). Wolfstetter \cite{wolfstetter01third} solved the equilibrium $k$th price auctions for a uniform distribution. Recently, Nawar and Sen \cite{nawar2018kth} represented the solution of Monderer and Tennenholtz's characterization equation as a finite series involving Catalan numbers. With their representation, they provided a closed form solution of the unique symmetric, increasing equilibrium of a $k$th price auction for a second degree polynomial distribution.

\vskip 5mm

In this paper, we analysis Monderer and Tennenholtz's characterization equation by using a method involving inverse of distribution functions. We provide a new representation of the equilibrium bid function of $k$th price auction with this representation. For applications, we extend Nawar and Sen's results and provide a closed form solution of a $k$th price auction for polynomial distribution, fat tail distribution and exponential distributions.

\vskip 5mm

After recalling the framework of the problem in Section 2, we prove our main result in Section 3. Then, in Section 4,  we compare our result with those of Nawar and Sen. Finally in Section 5 we provide a closed-form solution of the equilibrium bid function for polynomial distribution, exponential distribution and a class of fat tail distributions.

\section{Notations and assumptions}\label{S2}
In this section we present our assumptions and recall the result on the uniqueness of the equilibrium strategies provided by Monderer and Tennenholtz. 
Consider a $k$th price auction with $n$ bidders, where the highest bidder wins, and pays only the $k$th highest bid. Let $k \geq 2$ and $n > k$. We make the  following assumptions as in \cite{RS1981}:
\begin{enumerate}
    \item The valuations $X_i, i = 1,\cdots, n$ of the bidders are independent and identically distributed with distribution function $F$.
\item The distribution function $F$ is with values in $I$ where  $I= [0,1]$ or $I=\mathbb{R}^+$.
\end{enumerate}
We also assume that:
\begin{itemize}
\item[(A)]  $F$ is $k-2$ times continuously differentiable and the density function $f:= F'$ satisfying $ \forall x \in I$, $f(x) >0$.
\end{itemize}
We remark that under the assumption of the density function $f$, the quantile function $Q:=F^{-1}$ exists and is well-defined on $(0,1)$.
Note that for an analysis of the 3rd price auction in the literature, the existence and the continuity of the density function are often assumed. It is thus natural to assume $(A)$ holds for the case of general $k$th price auctions.  
We denote by $\beta_i$ the strategy of the bidder $i$ which determines its bid for any value. A strategy profile for $n$ bidders is given by $(\beta_1,...,\beta_n)$. 
A strategy profile is symmetric if $\beta_i$ are all equal to a common strategy $\beta$. A symmetric strategy is increasing if $\beta$ is an increasing function. The equilibrium of a $k$th price auction with a symmetric strategy profile is characterized by the following theorem: 
\begin{thm}\label{thm:MT}[Monderer and Tennenholtz\cite{monderer2000k}]
Let $\beta: [0,1]\mapsto \R^+$. A symmetric strategy profile with common strategy $\beta$ is an equilibrium of the $k$th price auction if and only if the following two conditions hold. \\
(E1) $\beta$ is an increasing function.\\
(E2) For all $x\in [0,1]$:
\begin{equation}\label{eq1}
    \int_{0}^x [x-\beta(y)]F(y)^{n-k}[F(x)-F(y)]^{k-3}f(y)dy = 0.
\end{equation}
Moreover \eqref{eq1} has at most one solution and for such a solution $\beta$ is differentiable for all $x\in Supp(f)$, here $Supp(f)$ denotes the support of the distribution with density $f$.
\end{thm}
According to Theorem \ref{thm:MT}, for a given $F$, if we can compute $\beta$ and show that $\beta$ is differentiable and increasing, then $\beta$ is the unique equilibrium bid function. If the equilibrium bid function $\beta$ exists,  $\beta(X)$ is the random variable representing the equilibrium strategy of bidders with valuation $X$. Furthermore, if $\beta$ is strictly increasing, together with differentiability of $\beta$ and the assumption that $F$ has a density, we can deduce that $\beta(X)$ also has a continuous density function. We denote $\hat F$ the distribution function of $\beta(X)$ (that is, $\hat F(\beta(x)) = F(x)$) and $\hat Q:= \hat F^{-1}$ the quantile function. 

\section{Analysis of equilibrium}
Here we present our main result. We give a closed form solution of equation \eqref{eq1} for $k\geq 3$. With this solution, we are able to find a closed form expression of the bid function for some non linear distributions.   The key idea is, instead of working directly with the distribution function $F$ as in the literature, we use the quantile function $Q$, which can largely simplify the calculus and give a better insight. 
\begin{thm}\label{thm1}
Assume that $\beta:[0,1]\mapsto \R^+$ is an increasing function and $\beta(X)$ has an increasing distribution function with a continuous density. Denote $\tilde x = F(x)$ and $\gamma(\tilde x):= Q(\tilde x)\tilde x^{n-2}$. Then \eqref{eq1} has a unique solution given by
\begin{equation}\label{eq3}
    \beta(x) = \frac{\gamma^{(k-2)}(\tilde x)}{(k-2)!\binom{n-2}{k-2}F(x)^{n-k}},
\end{equation}
where the superscript $(t)$ is the $t^{th}$ order derivative with respect to $\tilde x$.
\end{thm}
\begin{proof}
Denote
\begin{equation}\label{eq4}
    s(x): = x\int_0^x[F(x)-F(y)]^{k-3}F(y)^{n-k}f(y)dy,
\end{equation}
\begin{equation}\label{eq5}
    w(x): = \int_0^x\beta(y)[F(x)-F(y)]^{k-3}F(y)^{n-k}f(y)dy.
\end{equation}
Note that \eqref{eq1} can be written as $s(x) = w(x)$.\\
Since $F'(y) = f(y)$, making  the transformation $z=F(y)/F(x)$ in \eqref{eq4} we have
\begin{equation}  
    s(x) = xF(x)^{n-2}\int_0^1(1-z)^{k-3}z^{n-k}dz = xF(x)^{n-2}B(n-k+1,k-2),
\end{equation}
where $B(c,d)$ is the beta function with $B(c,d) = (c-1)!(d-1)!/(c+d-1)!$ for positive integers $c$,$d$. It follows that 
\begin{equation}\label{eq6}
    s(x) = \frac{xF(x)^{n-2}}{(k-2)\binom{n-2}{k-2}}.
\end{equation}
Make the transformation $z = F(y)$ in \eqref{eq5}. As $\beta$ is increasing, we have $F(y) = \hat F(\beta(y))$, so that 
$$\beta(y)= \hat{F}^{-1}(F(y)) = \hat Q(F(y)) = \hat Q(z).$$
Thus,
\begin{equation}\label{eq7}
     w(x) = \int_0^{F(x)}\hat Q(z)[F(x)-z]^{k-3}z^{n-k}dz.
\end{equation}
As $\tilde x = F(x)$, we have $x = Q(\tilde x)$. Denote
$$
S(\tilde x) := s(Q(\tilde x)) \ \text{and} \ W(\tilde x):= w(Q(\tilde x)).
$$
Also recall that $\gamma(\tilde x):= Q(\tilde x)\tilde x^{n-2}$. Then from \eqref{eq6}:
\begin{equation}\label{eq9}
    S(\tilde x) = s(Q(\tilde x))=s(x) = \frac{\gamma(\tilde x)}{(k-2)\binom{n-2}{k-2}}.
\end{equation}
By \eqref{eq7}:
\begin{equation}\label{eq10}
    W(\tilde x):= w(Q(\tilde x)) = w(x) = \int_0^{\tilde x} \hat Q(z)(\tilde x -z)^{k-3}z^{n-k}dz.
\end{equation}
According to \eqref{eq1}, $S(\tilde x) = W(\tilde x)$. Noticing that assumption (A) ensures that $S(\tilde x)$ is $k-2$ times continuously differentiable, we have:
\begin{equation}\label{eqstar}
S^{(k-2)}(\tilde x)= W^{(k-2)}(\tilde x).
\end{equation}
$S^{(k-2)}(\tilde x)$ is known by \eqref{eq9}. To find $W^{(k-2)}(\tilde x)$, apply Lemma \ref{lemma1} of the Appendix. Taking $m = n-k$ and $r=k-3$ in Lemma \ref{lemma1}, it follows that 
$$W^{(k-2)}(\tilde x) = (k-3)!\hat Q(\tilde x)\tilde x^{n-k}.$$
Together with \eqref{eq9} and \eqref{eqstar}, we have:
\begin{equation}\label{eq11}
\frac{\gamma^{(k-2)}(\tilde x)}{(k-2)\binom{n-2}{k-2}} = (k-3)!\hat Q(\tilde x)\tilde x^{n-k}.
\end{equation}
As $\hat F(\beta(x)) = F(x) = \tilde x$, we have $\hat Q( \tilde x) = \hat F^{-1}(\tilde x) = \beta(x)$. Using this, \eqref{eq3} follows from \eqref{eq11}.
\end{proof}
\begin{rque}
By the Leibniz rule for derivation, 
\begin{align*}
\gamma^{(k-2)}(\tilde x)& = \sum_{i=0}^{k-2}\binom{k-2}{i} (\tilde{x}^{n-2})^{(k-2-i)} Q^{(i)}(\tilde x)\\
& = \sum_{i=0}^{k-2}\binom{k-2}{i}\frac{(n-2)!}{(n-k+i)!} \tilde{x}^{n-k+i}Q^{(i)}(\tilde x)
\end{align*}
Recall that $\tilde x = F(x)$, from equation \eqref{eq3}, we have:
\begin{align}\label{eqq4}
    \beta(x)& = \sum_{i=0}^{k-2}\binom{k-2}{i}\frac{(n-k)!}{(n-k+i)!} \tilde{x}^{i}Q^{(i)}(\tilde x)\nonumber\\
    & = Q(\tilde x) +\frac{k-2}{n-k+1} Q'(\tilde x) + \sum_{i=2}^{k-2}\binom{k-2}{i}\frac{(n-k)!}{(n-k+i)!} \tilde{x}^{i}Q^{(i)}(\tilde x)
\end{align}
Since $Q(\tilde x) = F^{-1}(\tilde x) = x$ and $Q'(\tilde x) = 1/F'(x) = 1/f(x)$, for $k\geq 4$ we have
$$\beta(x) =x+\frac{k-2}{n-k+1}\frac{F(x)}{f(x)} + \sum_{i=2}^{k-2}\frac{\binom{k-2}{i}F(x)^i Q^{(i)}(\tilde x)}{(n-k+1)...(n-k+i)} =x+\frac{k-2}{n-k+1}\frac{F(x)}{f(x)}+O\left(\frac{1}{n^2}\right).$$
    This result coincides with the result of Wolfsteller \cite{wolfstetter01third} in $O(\frac{1}{n^2})$. Moreover it agrees with the expression in proposition 3 of  Gadi Fibich and Arieh Gavious's work in \cite{FG2010}.
\end{rque}



\section{Comparison with Nawar and Sen's result}
Applying the revenue equivalence principle for expected payment of a bidder in a $k$th price auction with $n$ bidders, Nawar and Sen (2018)\cite{nawar2018kth} have obtained the following expression of $\beta(x)$:
\begin{equation}
    \beta(x) = \frac{\psi_{k-1}(x)}{(k-2)!\binom{n-2}{k-2}F(x)^{n-k}},
    \label{nawar0}
\end{equation}

where 

\begin{equation}
\psi_{0}(x)=\int_0^x yF(y)^{n-2}f(y)dy \text{ and } \psi_{t+1}(x)=\frac{\psi_{t}'(x)}{f(x)} \text{ for } t=0,1,...
\label{nawar1}
\end{equation}

By making the transformation $z=F(y)$ we have $\psi_0(x)=\int_0^{ \tilde x} Q(z)z^{n-2}dz$.

Thus $$\frac{d \psi_0(x)}{d \tilde x}=Q(\tilde x) \tilde x^{n-2}=\gamma(\tilde x).$$

Also note that as $\tilde x=F(x)$, we have $\frac{d \tilde x}{dx}=F'(x)=f(x)$. By \eqref{nawar1}, we have 

$$ \psi_1(x)=\frac{d \psi_0(x)}{dx}\frac{1}{f(x)}=\frac{d \psi_0(x)}{d \tilde x} \frac{d \tilde x}{dx} \frac{1}{f(x)}=\gamma(\tilde x).$$

Applying the iterative definition of \eqref{nawar1} again, we have:

$$ \psi_2(x)=\frac{d \psi_1(x)}{dx}\frac{1}{f(x)}=\frac{d \psi_1(x)}{d \tilde x} \frac{d \tilde x}{dx} \frac{1}{f(x)}=\frac{ d \gamma(\tilde x)}{d \tilde x}=\gamma^{(1)}(\tilde x).$$

By induction, if $\psi_t(x)=\gamma^{(t-1)}(\tilde x)$, then

$$ \psi_{t+1}(x)=\frac{d \psi_t(x)}{dx}\frac{1}{f(x)}=\frac{d \psi_t(x)}{d \tilde x} \frac{d \tilde x}{dx} \frac{1}{f(x)}=\frac{ d \gamma^{(t)}(\tilde x)}{d \tilde x}.$$

This shows that $\psi_t(x)=\gamma^{(t-1)}(\tilde x)$ for $t=1,2,....$ So $\psi_{k-1}(x)=\gamma^{(k-2)}(\tilde x)$ and the expression \eqref{nawar0} coincides with our expression \eqref{eq3}. Therefore, our expression \eqref{eq3} is an equivalence representation of Nawar and Sen's result. Instead of expanding $\psi$ with series about Catalan numbers, we compute it with the quantile function. From this expression \eqref{eq3}, it is easy to establish  the equilibrium for some non linear distributions, and we will detail them in the next section.

\section{Examples of Equilibrium for some distributions}
\subsection{Equilibrium for non linear distribution}
 In this section we study the equilibrium bid function for some non-linear distributions. First of all, as a corollary of theorem \ref{thm1}, we provide sufficient conditions for the existence and uniqueness of the equilibrium bid function. Then, we will provide a closed-form solution of the equilibrium bid function for polynomial distribution, exponential distribution, a class of fat tail distributions.
 
 According to theorem \ref{thm:MT}, for some distribution $F$, if one can show that $\beta$ found by \eqref{eq3} is an strictly increasing function, it will follow that the symmetric strategy profile with common strategy $\beta$ is an equilibrium of the $k$th price auction. 
\begin{cor}\label{sufficient condition}
Consider a $k$th price auction where each $X_i$ is i.i.d. on the interval $[0,1]$ with quantile function $Q$. Assume that it holds
\begin{equation}
    \forall i \in [|1,k-1|], \,  \forall \tilde x \in [0,1), \,  Q^{(i)}( \tilde x) > 0.
    \label{derivate1}
\end{equation}
Then the existence and uniqueness of equilibrium bid function is given by \eqref{eq3} and can be rewritten as 
\begin{equation}\label{eqq3}
    \beta(x) = x + \sum_{i=1}^{k-2}\binom{k-2}{i}\frac{(n-k)!}{(n-k+i)!} \tilde{x}^{i}Q^{(i)}(\tilde x),
\end{equation}
with $\tilde x = F(x)$.
\end{cor}
\begin{proof}
Equation \eqref{eqq3} is direct from equation \eqref{eqq4}.
By equation \eqref{derivate1}, we deduce that for $i = 0,\cdots,k-2$, $Q^{(i)}(\tilde x)$ is strictly increasing with respect to $\tilde x$. Together with the fact that $\tilde x = F(x)$ is increasing strictly with respected to $x$, according to equation \eqref{eqq4}, we deduce that $\beta(x)$ is strictly increasing.
 The conclusion follows.
\end{proof}

For some distributions, the condition $Q^{(i)}(\tilde{x})>0$ for all $\tilde x \in [0,1)$ is easy to check but for some it is not. For many non-linear distributions, the quantile function $Q$ is analytic on $[0,1)$. According to the latter corollary, we provide another sufficient condition as follows, the condition $(P^+)$, which is easy to check.

{\bf Condition $(P^+)$} : The function $Q$ is analytic on $[0,1)$ with positive coefficient in the representation of power series. More precisely, there exist positive $\alpha_i$, such that for $x\in [0,1)$, 
$$Q(\tilde x)= \sum_{i=0}^\infty \alpha_i \tilde{x}^i.$$
The equality is defined in the sense of power series, reader can refer to any power series literature for a complete justification of technical convergence details. From condition $(P^+)$, deriving the equation successively, it is easy to see that equation \eqref{derivate1} holds. According to corollary~\ref{sufficient condition}, $\beta$ is an equilibrium bid function. To check the condition $(P^+)$, we only need to calculate the Taylor expansion of the quantile function $Q$ on $0$, then check the sign of the Taylor coefficient. Here are some examples of applications of corollary~\ref{sufficient condition}.

\begin{exple}[Exponential distribution.]
Let $F(x) := 1-e^{-\lambda x }$ for $\lambda>0$ and $x\in \Rr^+$. The equilibrium bid function is given by \eqref{eq3}. 
In fact $Q(x) := \frac{-1}{\lambda} \ln (1-x)$. Moreover $Q^{(i)}(x) = \frac{i!}{\lambda(1-x)^i}$ for $i\in [|1,k-1|]$, which is strictly positive on $[0,1)$. According to Corollary \ref{sufficient condition}, the equilibrium bid function $\beta(x)$ has the expression:

$$\beta(x)=  x +\frac{1}{\lambda} \sum_{i=1}^{k-2}\binom{k-2}{i} \frac{(n-k)!}{(n-k+i)!}\tilde x ^{i} \frac{i!}{^(1-\tilde x)^i} ,$$
where $\tilde x =F(x)$.
\end{exple}

\begin{exple}[Fat tail distribution.]
Let $F(x) := 1-\frac{1}{x^c}$ for some $c>0$ and for $x\in \Rr^+$. The equilibrium bid function is given by \eqref{eq3}. 
In fact $Q(x):=\frac{1}{(1-x)^c}$. Moreover $Q^{(i)}(x) = \frac{c(c+1)...(c+i-1)\tilde x^i}{(1- \tilde x)^{c+i}}$ for $i\in [|1,k-1|]$, which is strictly positive on $[0,1)$. According to Corollary \ref{sufficient condition}, the equilibrium bid function $\beta(x)$ has the expression:

$$\beta(x)= x+ \sum_{i=1}^{k-2} \binom{k-2}{i} \frac{(n-k)!}{(n-k+i)!}\frac{c(c+1)...(c+i-1)\tilde x^i}{(1- \tilde x)^{c+i}}, $$

where $\tilde x =F(x)$.
\end{exple}

\begin{thm}
Consider a $kth$ price auction where each $X_i$ is iid on the interval $[0,1]$, with distribution function $F(x) = x^\alpha$ where $\alpha >0$. Then there is a unique symmetric equilibrium. The equilibrium common strategy $\beta:[0,1]\mapsto \R^+$ is 
\begin{equation}\label{eqpoly}
    \beta(x) = \frac{\Gamma(n-k+1) \Gamma(n-1+1/\alpha)}{\Gamma(n-1)\Gamma(n-k+1+1/\alpha)} x,
\end{equation}
where $\Gamma$ is the Gamma function.
In particular, if $\alpha = \frac{1}{m}$ where $m$ is a positive integer, 
$$\beta(x) = \frac{(n-k+m+1)...(n-2+m)}{(n-k-+1)...(n-2)}x.$$
\end{thm}
\begin{proof}
To prove this, we find $\beta$ using theorem \ref{thm1} and show that it is a strictly increasing function of $x$. 
As  $\tilde x = F(x) = x^\alpha$, it follows $Q(\tilde x) = F^{-1}(\tilde x) = \tilde x^{1/\alpha} = x$. Then 
$$\gamma(\tilde x)= Q(\tilde x)\tilde x^{n-2} = \tilde x^{n-2+1/\alpha}.$$
Therefore,
\begin{align*}
    \gamma^{(k-2)}(\tilde x) &= (n-2+1/\alpha)(n-3+1/\alpha)...(n-k+1+1/\alpha)\tilde x^{n-k+1/\alpha}\\
    & = \frac{\Gamma(n-1+1/\alpha)}{\Gamma(n-k+1+1/\alpha)}F(x)^{n-k}x.
\end{align*}
Together with \eqref{eq3}, \eqref{eqpoly} follows.
\end{proof}

\section{Acknowledgments}
 
the reviewers had a key role in the conception of this article and we are grateful for the reviewers for the very helpful comments.

\section{Appendix}
\begin{lemma}\label{lemma1}
Consider a real valued bounded function $\hat{Q}:\R\mapsto [0,1]$. For positive integer $r$ and positive real number $m$, let $A_r(u,z):= \hat Q(z)(u-z)^r z^m$ and $H_r(u):=\int_0^u A_r(u,z) dz$. Then the $(r+1)^{th}$ derivative of $H_r(u)$ is 
\begin{equation}
    H_r^{(r+1)}(u) = r!\hat{Q}(u)u^m.
\end{equation}
\end{lemma}
\begin{proof}
By the Leibniz rule \cite{rudin1964principles} of differentiating an integral, if $H(u):= \int_{l_0(u)}^{l_1(u)} A(u,z) dz$, under assumption of integrability of $\partial_u A(u,z)$, it holds:
$$H'(u) = \left[A(u,l_1(u))l_1'(u) - A(u,l_0(u))l_0'(u)\right] +\int_{l_0(u)}^{l_1(u)} \partial_u A(u,z) dz.$$
It is easy to check the integrability of $\partial_u\hat Q(z)(u-z)^r z^m$, thus taking $l_0(u)=0, l_1(u) = u$ in the previous equation:
\begin{equation}\label{eq30}
    H_r'(u) = A_r(u,u) +\int_0^u \partial_u A_r(u,z) dz.
\end{equation}
For $r=0$, $A_0(u,u) = \hat{Q}(u)u^m$ and $\partial_u A_0(u,z)=0$. For $r\geq 1$, $A_r(u,u) = 0$ and $\partial_u A_r(u,z) = \hat{Q}(z)r(u-z)^{r-1}z^m = rA_{r-1}(u,z)$. Therefore \eqref{eq30} implies
$$H_r'(u) = rH_{r-1}(u) \text{\ for\ } r\geq 1 \text{\ and\ } H_0'(u) = \hat{Q}(u)u^m.$$ 
Thus for $t\leq r$:
$$H_r^{(t)}(u) = r(r-1)...(r-t+1)H_{r-t}(u),$$
so $H_r^{(r)}(u) = r!H_0(u)$ and therefore $H_r^{(r+1)}(u) = r!H_0'(u) = r!\hat{Q}(u)u^m.$
\end{proof}

\begin{thebibliography}{10}

\bibitem{FG2010}
G.~Fibich and A.~Gavious.
\newblock Large auctions with risk-averse bidders.
\newblock {\em International Journal of Game Theory}, 39(3):359--390, 2010.


\bibitem{PK99}
Paul Klemperer.
\newblock {Klemperer, P (1999) Auction theory: a guide to the literature}.
\newblock Journal of Economic Surveys, 13(3):227-286

\bibitem{KRISHNA20101}
Vijay Krishna.
\newblock {\em Auction Theory (Second Edition)}.
\newblock Academic Press, San Diego, second edition, 2010.

\bibitem{EM04}
Eric Maskin.
\newblock The unity of auction theory.
\newblock {\em Journal of Economic Literature}, 42(4):1102--1115, 2004.

\bibitem{monderer2000k}
Dov Monderer and Moshe Tennenholtz.
\newblock $k$-price auctions.
\newblock {\em Games and Economic Behavior}, 31(2):220--244, 2000.

\bibitem{Roger1981}
Roger~B. Myerson.
\newblock Optimal auction design.
\newblock {\em Mathematics of Operations Research}, 6(1):58--73, 1981.


\bibitem{RS1981}
John~G. Riley and William~F. Samuelson.
\newblock Optimal auctions.
\newblock {\em The American Economic Review}, 71(3):381--392, 1981.

\bibitem{vickrey1961counterspeculation}
William Vickrey.
\newblock Counterspeculation, auctions, and competitive sealed tenders.
\newblock {\em The Journal of Finance}, 16(1):8--37, 1961.

\bibitem{WILSON92}
Robert Wilson.
\newblock Wilson R (1992) Strategic analysis of auctions. In: Auman R and Hart S (eds.)
\newblock Handbook og Game Theory with Economic Applications, Vol. 1, Elsevier Science Publishers (North-Holland), pp. 227-279

\bibitem{wolfstetter01third}
E.~Wolfstetter.
\newblock Third- and higher-price auctions.
\newblock In M.~Braulke and S.~Berninghaus, editors, {\em Beiträge zur
  {M}ikro- und zur {M}akroökonomik: {F}estschrift für {H}ans-{J}ürgen
  {R}amser}, pages 495--503. Springer Verlag, 2001.
  
\bibitem{nawar2018kth}
Nawar A-H and Sen D (2018) .
\newblock $k$th price auctions and Catalan numbers.
\newblock Economics Letters, volume 172, pages 69--73, Elsevier, 2018.

\bibitem{rudin1964principles}
Rudin, Walter and others.
\newblock Principles of mathematical analysis.
\newblock Volume 3,  year 1964, McGraw-hill New York.



\end{thebibliography}

\end{document}